\documentclass[11pt]{article}
\usepackage[utf8]{inputenc}
\usepackage[T1]{fontenc}
\usepackage[english]{babel}
\usepackage[margin=1.15in]{geometry}
\usepackage{parskip}
\usepackage{amsmath}
\usepackage{amsthm}
\usepackage{amsfonts}
\usepackage{graphicx}
\usepackage{float}
\usepackage{pgfplots}
\pgfplotsset{compat=1.14, set layers}
\usepackage{bm}
\usepackage{mathtools}
\usepackage{subfig}
\usepackage{enumitem}

\usepackage{lipsum}
\usepackage{setspace}
\usepackage{caption}
\usepackage{xspace}
\definecolor{refcolor}{rgb}{0.23, 0.27, 0.29}
\usepackage[breaklinks,colorlinks,urlcolor={refcolor},citecolor={refcolor},linkcolor={refcolor}]{hyperref}
\newtheorem{theorem}{Theorem}[section]
\newtheorem{lemma}[theorem]{Lemma}

\newtheorem{definition}[theorem]{Definition}
\newtheorem{remark}[theorem]{Remark}
\newtheorem{claim}[theorem]{Claim}

\newcommand{\arr}{\xrightarrow{}}
\newcommand{\larr}{\xleftarrow{}}
\newcommand{\local}{\textsf{LOCAL}\xspace}
\newcommand{\mpc}{\textsf{MPC}\xspace}
\newcommand{\rc}{$\varrho$\xspace}

\begin{document}

\begin{center}
    {\huge \bf Coloring Trees in Massively Parallel Computation} \\
\vspace{1cm}

\begin{minipage}[H]{5cm}
\begin{center}
    {\large \bf Rustam Latypov\footnotemark} \\
    {\large Aalto University} \\ \href{mailto:rustam.latypov@aalto.fi}{\texttt{rustam.latypov@aalto.fi}}
\end{center}
\end{minipage}
\begin{minipage}[H]{5cm}
\begin{center}
    {\large \bf Jara Uitto} \\
    {\large Aalto University} \\
    \href{mailto:jara.uitto@aalto.fi}{\texttt{jara.uitto@aalto.fi}}
\end{center}
\end{minipage}

\vspace{4mm}
\begin{center}
    \today \\ 
\end{center}

\vspace{1cm}
\begin{minipage}[H]{13.5cm}
\begin{center}
    \textbf{Abstract} \\ \vspace{4mm}
\end{center}
    We present $O(\log^2 \log n)$ time 3-coloring, maximal independent set and maximal matching algorithms for trees in the Massively Parallel Computation (\mpc) model. Our algorithms are deterministic, apply to arbitrary-degree trees and work in the low-space \mpc model, where local memory is $O(n^\delta)$ for $\delta \in (0,1)$ and global memory is $O(m)$. Our main result is the 3-coloring algorithm, which contrasts the randomized, state-of-the-art 4-coloring algorithm of Ghaffari, Grunau and Jin [DISC'20]. The maximal independent set and maximal matching algorithms follow in $O(1)$ time after obtaining the coloring. The key ingredient of our 3-coloring algorithm is an $O(\log^2 \log n)$ time adaptation of the rake-and-compress tree decomposition used by Chang and Pettie [FOCS'17], and established by Miller and Reif. When restricting our attention to trees of constant degree, we bring the runtime down to $O(\log \log n)$.
\end{minipage}
\end{center}

\thispagestyle{empty}
\footnotetext{Supported in part by the Academy of Finland, Grant 334238}

\vspace{8mm}
\section{Introduction}

When confronted with massive data sets that dwarf the memory capacity of any single machine, the need arises for robust distributed systems, where multiple machines coordinate in a decentralized fashion to efficiently process data in parallel. In the field of distributed algorithms, we study mathematical abstractions of distributed systems and the ultimate goal is to characterize the computational power of such systems. One of the earliest abstractions was the \local model by Linial \cite{linial} and in the last few decades, many other models have emerged such as \textsf{BSP} \cite{bsp}, \textsf{CONGEST} \cite{congest}, \textsf{CONGESTED CLIQUE} \cite{congestc3,congestc1,congestc2} and \mpc \cite{karloff,ref1,ref2,ref3,ref4,czumaj}. 

The Massively Parallel Computation (\mpc) model is a mathematical abstraction of modern data processing platforms such as MapReduce~\cite{dg04}, Hadoop~\cite{White:2012}, Spark~\cite{ZahariaCFSS10}, and Dryad~\cite{Isard:2007}. The low-space variant of the \mpc model is of particularly interest, as it captures the communication capabilities and the memory limitations of the real world.

\vspace{5mm}

\subsection{The \mpc Model}

In this model, we have $M$ machines who communicate in an all-to-all fashion. We focus on problems where the input is modeled as a graph with $n$ vertices, $m$ edges and maximum degree $\Delta$; we call this graph the \textit{input graph}. Each node has a unique ID of size $O(\log n)$ bits from a domain $\{1,2,\dots,N\}$, where $N = \text{poly}(n)$. Each node and its incident edges are hosted on a machine with $O(n^\delta)$ \textit{local memory} capacity, where $\delta \in (0,1)$ and the units of memory are words of $O(\log n)$ bits. When the local memory is bounded by $O(n^\delta)$, the model is called \textit{low-space} (or \textit{sublinear}). The number of machines is chosen such that $M = \Theta(m/n^\delta)$. For trees, where $m=\Theta(n)$, this results in $\Theta(n^{1-\delta})$ machines. For simplicity, we assume that each machine $i$ simulates one virtual machine for each node and its incident edges that $i$ hosts, such that the local memory restriction becomes that no virtual machine can use more than $O(n^\delta)$ memory. If a node cannot be stored on a single machine, as its degree is $\omega(n^\delta)$, one has to use some sort of a workaround. In this work we resort to having several smaller-degree copies of high degree nodes on separate machines. These machines can be arranged in a $n^\delta$-ary tree such that they can communicate in constant time and act as a single machine\footnote{As all of our computations consist of separable functions \cite{behnezhad2}, this workaround is sufficient.}.

During the execution of an \mpc algorithm, computation is performed in synchronous, fault-tolerant rounds. In each round, every machine performs some (unbounded) computation on the locally stored data, then sends/receives messages to/from any other machine in the network. Each message is sent to exactly one other machine specified by the sending machine. All messages sent and received by each machine in each round, as well as the output, have to fit into local memory. The time complexity is the number of rounds it takes to solve a problem. Upon termination, each node knows must know its own part of the solution.

The \textit{global memory} (or \textit{total memory}) use during the execution of an algorithm is the sum over the \textit{used} local memory over all machines. In this work, we restrict the global memory use to linear in the number of edges -- $O(m)$, which is the strictest possible as it is only enough to store a constant number of copies of the input graph. Note that if we were to allow superlinear $O(m^{1+\delta})$ global memory in constant-degree trees, many \local algorithms with complexity $O(\log n)$ could be exponentially sped up in the low-space \mpc model by applying the well-known graph exponentiation technique by Lenzen and Wattenhofer \cite{wattenhofer}. A crucial challenge that comes with the linear global memory restriction is that only a small fraction of $n^{1 - \delta}$ of the (virtual) machines can simultaneously utilize all of their available local memory.

\subsection{Related Works}

Fundamental graph problems such as coloring, MIS, maximal matching, minimum matching and minimum vertex cover have enjoyed a lot of attention in recent years in the low-space \mpc regime. For general graphs with maximum degree $\Delta$, Ghaffari and Uitto \cite{kaksikko} gave $\widetilde{O}(\sqrt{\log \Delta})$ time algorithms for MIS, maximal matching, $(1+\varepsilon)$-approximation of maximum matching, and 2-approximation of minimum vertex cover, which remain the best known. Only when considering special graph families, namely trees and other sparse graphs bounded by arboricity $\alpha$, the bound of $\widetilde{O}(\sqrt{\log \Delta})$ has been outperformed. Behnezhad et al. \cite{behnezhad2} gave MIS and maximal matching algorithms such that they first reduce problems in graphs with arboricity $\alpha$ to corresponding problems in graphs with maximum degree $\text{poly}(\alpha, \log n)$, and then by invoking the aforementioned algorithm of Ghaffari and Uitto \cite{kaksikko}, they arrive at MIS and maximal matching algorithms with time complexity $O(\sqrt{\log \alpha} \cdot \log \log \alpha + \log^2 \log n)$. Finally, a recent paper by Ghaffari et al. \cite{ghaffari} presents, to the best of our knowledge, the current state-of-the-art time complexities of $O(\sqrt{\log \alpha} \cdot \log \log \alpha + \log \log n)$ for MIS and maximal matching. As a secondary contribution, they provide an $O(\log \log n)$ time randomized algorithm for 4-coloring trees.

Regarding bounds, Ghaffari et al. \cite{focs} gave conditional lower bounds of $\Omega(\log \log n)$ for component-stable, low-space \mpc algorithms for constant approximation of maximum matching, constant approximation of vertex cover, and MIS. Their hardness results are conditioned on a widely believed conjecture in \mpc about the complexity of the connectivity problem.

\subsection{Preliminaries and Notations}

We work with undirected, finite, simple graphs $G = (V,E)$ with $n=|V|$ nodes and $m=|E|$ edges such that $E \subseteq [V]^2$ and $V \cap E = \emptyset$. In particular, we assume that $G$ is a graph with arboricity $\alpha=1$, i.e., a tree with $n$ nodes and $m=n-1$ edges. Let $\deg_G(v)$ denote the degree of a node $v$ in $G$ and let $\Delta$ denote the maximum degree of $G$. The distance $d_G(v,u)$ between two vertices $v,u$ in $G$ is the length of a shortest $v - u$ path in $G$; if no such path exists, we set $d_G(v, u) \coloneqq \infty$. The greatest distance between any two vertices in $G$ is the diameter of $G$, denoted by $\text{diam}(G)$. For a subset $S \subseteq V$, we use $G[S]$ to denote the subgraph of $G$ induced by nodes in $S$. Let $G^k$, where $k \in \mathbb{N}$, denote the $k$:th power of a graph $G$, which is another graph on the same vertex set, but in which two vertices are adjacent if their distance in $G$ is at most $k$. In the context of \mpc, $G^k$ is the resulting virtual graph after performing $\log k$ steps of graph exponentiation (Lemma \ref{lemma:exponentiation}).

For each node $v$ and for every radius $k \in \mathbb{N}$, we denote the $k$-hop (or $k$-radius) neighborhood of $v$ as $N^k(v) = \{ u \in V : d(v,u) \leq k\}$. The topology of a neighborhood $N^k(v)$ of $v$ is simply $G[N^k(v)]$. However, with slight abuse of notation, we sometimes refer to $N^k(v)$ both as the node set and the subgraph induced by node set $N^k(v)$. Neighborhood topology knowledge is often referred to as vision, e.g., node $v$ sees $N^k(v)$.

Let us also differentiate between the notions of the \textit{input graph} and the \textit{virtual graph} that we will use throughout this work. The input graph is the problem input, which we never change or remove throughout the execution of any algorithm.  The virtual graph on the other hand, is the toy graph that that is constructed during an algorithm in order to ease analysis. Nodes sharing an input edge are called neighbors and nodes sharing a virtual edge are called virtual neighbors. 

In the \mpc model, due to global communication, it is possible to collect neighborhoods exponentially faster that in the \local model using a by-now standard technique called graph exponentiation. The idea of graph exponentiation was first mentioned under the \textsf{CONGESTED CLIQUE} model by Lenzen and Wattenhofer \cite{wattenhofer}.

\begin{lemma}[Graph exponentiation] \label{lemma:exponentiation}
Assume that each node and its incident edges are hosted on a unique machine and that the $k$-hop neighborhood of each node contains at most $n^\delta$ nodes. All nodes $v$ in the graph learn their $k$-hop neighborhood in $O(\log k)$ \mpc rounds through graph exponentiation, when allowing $O(n^\delta)$ local memory and $O(n^{1+\delta})$ global memory.
\end{lemma}

\begin{proof}
Since there is a one-to-one correspondence between (virtual) machines and nodes, we can use the two terms interchangeably. Initially, all machines hold $N^1(\cdot)$ and the virtual graph is an identical copy of the input graph. In each iteration, all nodes connect their virtual neighbors with a virtual edge. In practice, connecting nodes (with a virtual edge) entails informing both nodes the ID of their new virtual neighbor. At the end of an iteration, all nodes drop any duplicate edges that they have in memory.

\begin{figure}[H]
	\centering
	\includegraphics[width=0.75\textwidth]{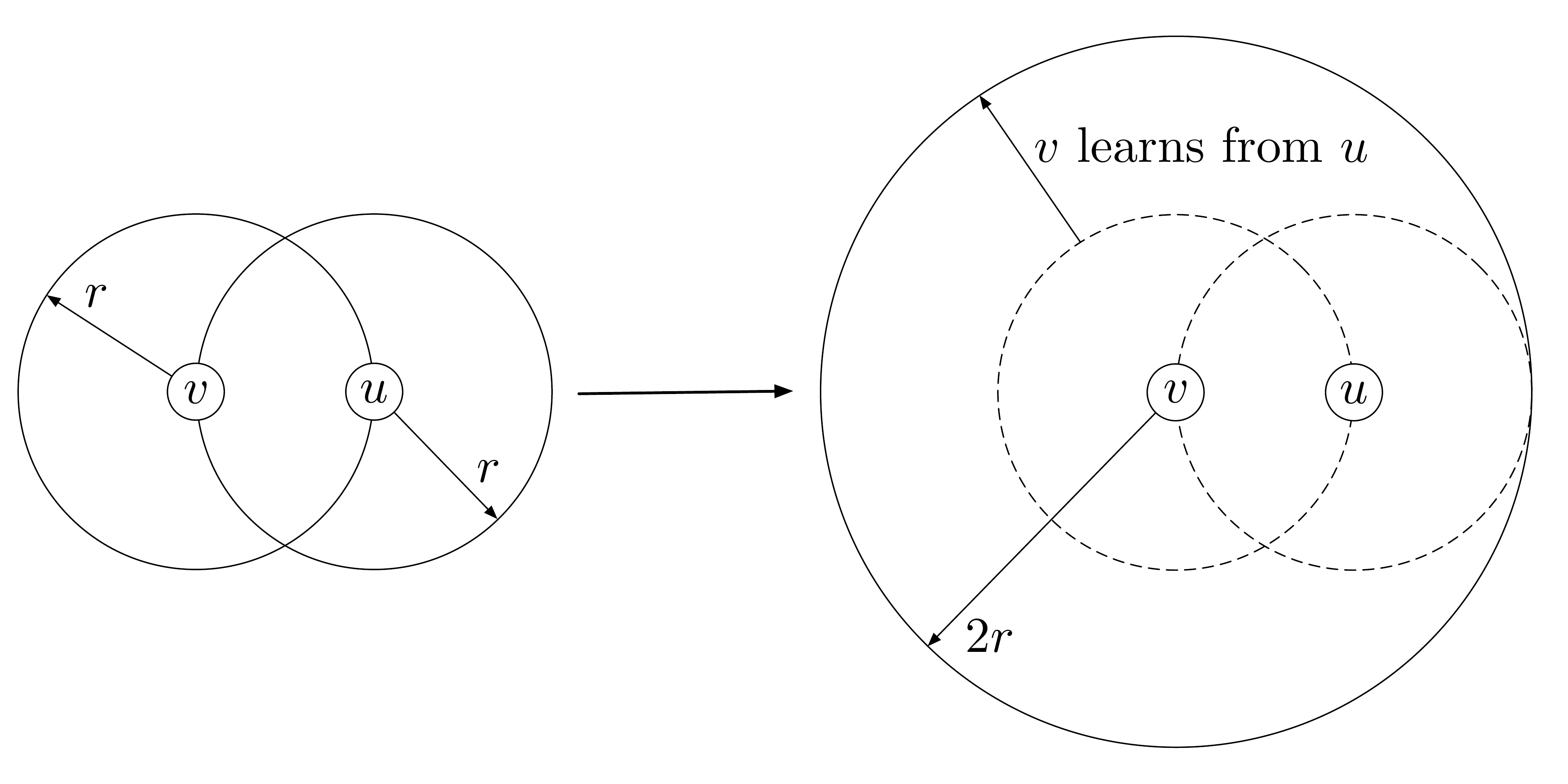}
	\caption{An illustration of node $v$ learning from nodes $u$ through graph exponentiation.}
	\label{fig:graphexpo}
\end{figure}

Since the distance (in the virtual graph) between any two nodes decreases by a factor of at least $3/2$ in each iteration, after $\log k$ iterations, each node is connected (in the virtual graph) to all nodes in its $k$-hop neighborhood, i.e., it has learned its $k$-hop neighborhood. Figure \ref{fig:graphexpo} illustrates $v$ learning from nodes $u$ at distance $r$. 
\end{proof}

Note that for trees, knowing your $k$-hop neighborhood is equivalent to knowing your $k$-hop neighborhood topology, since you can communicate with your $k$-hop neighborhood in constant time and the number of edges is bounded by the number of nodes. Fast neighborhood collection using graph exponentiation can be used to speed up \local algorithms due to the following well known equivalence of time complexity and neighborhood topology proven by, e.g., Kuhn et al. \cite{simu}.

\begin{remark} \label{rem:simu} There is a one-to-one correspondence between the time complexity of distributed algorithms in the \local model and the graph theoretic notion of neighborhood topology. Knowing your $t$-hop neighborhood allows instant simulation of any $t$ time \local algorithm.
\end{remark}

\subsection{Coloring, MIS and Matching}

A $k$-coloring of a graph $G=(V,E)$ is a special case of graph labeling such that each node $v \in G$ is assigned a color $c(v)$ such that $c : V \arr \{1,2,\dots,k\}$ for some $k$. A proper $k$-coloring is such that $c(v) \not= c(u)$ for all adjacent nodes $v,u \in V$. A maximal independent set (MIS) is a subset of nodes $I \subseteq V$ satisfying that (i) for each node $v \in I$, none of its neighbors are in $I$ and (ii) for each node $v \not \in I$, at least one of its neighbors is in $I$. A matching is a subset of edges $M \subseteq E$ such that each node has at most one incident edge in $M$. A matching is maximal if is not a proper subset of another feasible matching.

\section{Our Contributions and Roadmap}

Our main contribution is the following.

\vspace{1mm}
\fcolorbox{black}{black!2}{\begin{minipage}{0.98\textwidth}
		\textbf{Coloring.} There is a deterministic $O(\log^2 \log n)$ time 3-coloring algorithm for trees in the low-space \mpc model using linear global memory. Maximal independent set and maximal matching algorithms follow in $O(1)$ time. (Theorems \ref{thm:3coloring}, \ref{thm:mis}, \ref{thm:mm})
\end{minipage}}
\vspace{1mm}

Our deterministic result contrasts the current randomized state of the art for constant coloring, which is $O(\log \log n)$ by Ghaffari, Grunau, and Jin \cite{ghaffari}. We emphasize that the previous result is randomized and uses $4$ colors, which as opposed to 3 colors, enables the use of shattering techniques and coloring components with disjoint color palettes. Note that the $\Omega(\log \log n)$ conditional hardness results by Ghaffari, Kuhn and Uitto \cite{focs} imply that our algorithm is only a $\log \log n$ factor away from the optimal, at least when considering component-stable algorithms. 

In Theorems \ref{thm:bound3coloring}, \ref{thm:boundmis}, \ref{thm:boundmm}, when restricting our attention to constant-degree trees, we bring the runtime down to $O(\log \log n)$ for 3-coloring, MIS and maximal matching. 

Contrary to the introduction, we first show the constant-degree algorithm (Section \ref{sec:constalgo}), as it is only a simplified version of the general case and hence gives a nice introduction. After this, we give the general algorithm (Section \ref{sec:genalgo}). The key ingredient of our 3-coloring algorithm is performing a slightly modified version of the \local rake-and-compress tree decomposition of Chang and Pettie \cite{chang} almost exponentially faster in the low-space \mpc model. We note that the decomposition was established by Miller and Reif \cite{Miller}.

Even though our results apply to trees, as we will soon notice, an input graph may shatter into a forest during the coloring algorithms. This however is not a problem, since all arguments are local and apply to trees and forests alike. Note that this implies that the input graph could as well have been a forest. Nevertheless, we formulate our result in terms of trees to be in line with literature.

\section{Constant-Degree Algorithm} \label{sec:constalgo}

In this section, we will give an $O(\log \log n)$ time deterministic algorithm for 3-coloring constant-degree trees in the low-space \mpc model using $O(m)$ words of global memory. Our coloring algorithm consists of a decomposition stage (Section \ref{sec:bounddecompstage}) and a coloring stage (Section \ref{sec:boundcoloringstage}). 

\begin{theorem} \label{thm:bound3coloring}
	There is a deterministic $O(\log \log n)$ time 3-coloring algorithm for \emph{constant}-degree trees in the low-space \mpc model using $O(m)$ words of global memory.
\end{theorem}

\begin{proof}
	Follows from Sections \ref{sec:bounddecompstage} and \ref{sec:boundcoloringstage}.
\end{proof}

\subsection{Decomposition} \label{sec:decomp}

Let us introduce a slightly modified version of the \local rake-and-compress decomposition of Chang and Pettie \cite{chang}, which we will henceforth refer to as a \rc decomposition.

\begin{definition}[\rc decomposition] \label{def:rc} In a \rc decomposition, the node set $V$ of a graph is split into sets $V_1,V_2,\dots,V_L$ such that $V=V_1 \cup V_2 \cup \dots \cup V_L$ and $V_i \cap V_j = \emptyset$ for all $i \not= j$. Parameter $L$ is called the size of the decomposition and the sets are referred to as layers. The properties are:
	
\begin{itemize}
	\item each node $v \in V_i$ has at most two neighbors in subgraph $\bigcup V_j,~i\leq j  \leq L$, and
	\item a subgraph induced by the nodes in any single layer consist of singletons and paths of length at least $l$, for some fixed constant $l$.
\end{itemize}

\end{definition}

Note that a \rc decomposition is parameterized by $l$. Since it is a constant, the decomposition can be constructed using the following \local algorithm.

\begin{enumerate}
	\item For $i=1,2,\dots, O(\log n)$:
	\begin{enumerate}
		\item In step $i$, nodes of degree-1 and nodes that belong to a path of length $\geq l$ place themselves in layer $i$.
	\end{enumerate}
\end{enumerate}

The runtime of $O(\log n)$ is proven by Chang and Pettie \cite{chang}[Section 3.9] by showing that at least a constant $1/2(l+1) \geq 1/4l$ fraction of nodes are eliminated in each step. Observe that this results in the size of the decomposition also being $O(\log n)$. For the remainder of this work, we are going to refer to one step of this \local algorithm as a \textit{peeling step}. The intuition being that we peel off the lowest layer of the decomposition.

\subsection{Decomposition stage} \label{sec:bounddecompstage}

For the purpose of our work, it is enough to consider a \rc decomposition with $l=3$. Note that this decision results in each node having at most one neighbor in a strictly higher layer. When working in the low-space \mpc model, we want to obtain the \rc decomposition exponentially faster than in the \local model. In constant-degree graphs, this is quite straightforward to achieve with the following \mpc algorithm.

Recall that at least a constant $1/4l$ fraction of nodes are eliminated in each step of the \local algorithm. In other words, there are at most a constant $1-1/4l$ fraction of nodes left in the graph after each step. Set constant $c \larr \text{argmin}_c \{c : (1-1/4l)^c < 1/\Delta \}$, and observe that since $\Delta$ and $l$ are constants, $c$ is also constant.

\begin{enumerate}
	\item For $i=1,2,\dots, \log \log n^\delta$:
	\begin{enumerate}
		\item In phase $i$, simulate $2^{i-1}$ peeling steps $2c$ times, and then perform one graph exponentiation step.
	\end{enumerate}
	\item Simulate $O(\delta \log n)$ peeling steps.
\end{enumerate}

The correctness of the algorithm follows from the fact that one peeling step corresponds to one step of the very simple \local algorithm described previously. Let us argue why we can simulate multiple peeling steps in one \mpc round. Performing one peeling step in the \local model takes $l$ time. Hence, knowing your $x$-hop neighborhood in \mpc enables the simulation of $x/l$ peeling steps. Since $l$ is only a constant, we can implicitly assume that the simulation is performed $l$ times, effectively simulating exactly $x$ peeling steps and increasing the runtime only by a constant factor. 

Let us analyze the runtime of the algorithm. Step 1 consists of phases, in which we simulate a number of peeling steps and perform one graph exponentiation step. Each phase takes constant time: due to graph exponentiation and the fact that $\Delta$ is constant, in the beginning of a phase $i$, all nodes see their $2^{i-1}$-neighborhoods, and $2c$ is constant. If the graph is not empty after $\log \log n^\delta$ phases, we proceed to Step 2. By now, all nodes see their $(\delta \log n)$-hop neighborhoods and are able to simulate $\delta \log n$ peelings steps in one \mpc round. Recall that the size of the decomposition is $O(\log n)$ and hence, after simulating $\delta \log n$ steps $O(1/\delta)$ times, the graph must be empty. We conclude that the runtime is $O(\log \log n)$.

Now for the memory analysis. At most, all nodes will see their $(\delta \log n)$-hop neighborhoods. Since $\Delta$ is constant, said neighborhoods will contain at most $n^\delta$ nodes (and edges), and local memory $O(n^\delta)$ is not violated. During Step 1, in the beginning of a phase $i$, all nodes see their $2^{i-1}$-hop neighborhoods. After simulating $2^{i-1}$ peeling steps $2c$ times, there are at most 
\begin{align*}
	n \cdot (1-1/4l)^{2^{i-1} \cdot 2c} < n / \Delta^{2^{i}}
\end{align*}

nodes left in the graph. The exponentiation step of phase $i$ requires at most $\Delta^{2^{i}}$ memory for each node. Hence, we use at most $O(m)$ global memory in each phase.

\subsection{Coloring stage} \label{sec:boundcoloringstage}

Consider having the \rc decomposition described in Definition \ref{def:rc} with $l=3$, which we get from the previous section. We start by performing two preprosessing steps. First, we color each layer (or more precisely, the subgraph induced by the nodes in each layer) of the decomposition separately, in parallel, in $O(\log^* N)=O(\log^* n)$ time by running Linial's $O(\Delta^2)$-coloring algorithm \cite{linial}. Recall that $N=\text{poly}(n)$ denotes the maximum ID in our ID space. Since each node has at most two neighbors in the same layer, running Linial's algorithm results in a $O(1)$-coloring within each layer. Since we have a proper constant coloring within each layer, we can perform the following. In each round, all nodes with the highest color among neighbors in the same layer recolor themselves with the smallest color such that a proper coloring is preserved within each layer. Clearly, one color is eliminated in each round and since $\Delta \leq 2$ within each layer, we achieve a 3-coloring within all layers in a constant number of rounds. We can think of this first pre-processing step as coloring each node with its temporary color, and then later recoloring each node with its final color.

Secondly, we partially orient the graph and call this orientation the \textit{dependency orientation}. Recall that when $l=3$, each node has at most one neighbor in a strictly higher layer. If a node has a neighbor in a strictly higher layer, it orients the edge outwards. We say that $v$ depends on neighbor $u$ if the edge connecting them is directed from $v$ towards $u$. We refer to all nodes that $v$ recursively depends on as the \textit{dependency path} of $v$. 

Since the first preprocessing step guarantees proper 3-colorings within each layer, color conflicts can only occur along the layer boundaries, i.e., only nodes with a dependency paths may require further attention (the temporary color of all other nodes are also their final color). It is straightforward to realize that any node $v$ with a dependency path can recolor itself with its final color if it sees its whole dependency path. Also, observe that since there are $O(\log n)$ layers in the \rc decomposition, each dependency path is of length at most $O(\log n)$. In our algorithm, the idea is for each node to learn its dependency path using a modified version of graph exponentiation, called \textit{directed exponentiation}.

\begin{definition}[Directed exponentiation]
	Directed exponentiation functions identically to regular graph exponentiation except for the following exceptions. Nodes only keep track of outwards oriented edges, resulting in nodes learning \emph{directed neighborhoods} during exponentiation. In each round of directed exponentiation, node $v$ informs nodes $u$ in its directed neighborhood that it wants to be connected with a virtual edge with all nodes in the directed neighborhood of $u$.
\end{definition}

Since all nodes without a dependency path are already colored with their final color, consider performing the following \mpc algorithm only for nodes with a dependency path. 

\begin{enumerate}
	\item Freeze the $\log \log n$ lowest layers. 
	\item Perform $\log \log n^\delta$ directed graph exponentiation steps. 
	\item All nodes that see their whole dependency path simulate its recoloring and using that information, compute their final color. Perform $O(1/\delta)$ such simulations. 
	\item Unfreeze the $\log \log n$ lowest layers and sequentially recolor their dependency paths.
\end{enumerate}

Correctness follows from the fact that each node can recolor itself with its final color when seeing its whole dependency tree. The runtime follows from the fact that we only perform $\log \log n^\delta$ directed exponentiation steps, recolor the dependency paths in the current highest $\delta \log n$ highest layers $O(1/\delta)$ times, and process $\log \log n$ lowest layers sequentially. 

Let us analyze the memory usage of our algorithm. By freezing the lowest $\log \log n$ layers, we allocate $\log n$ global memory for each node. Since the length of a dependency path stored by each node during the algorithm is at most $\delta \log n$, we do not violate local nor global memory. Note that the sequential recoloring of unfrozen layers does not require memory. An annoying detail associated with directed exponentiation is that even though each node $v$ depends on at most one node, multiple nodes may depend on $v$. Hence, during directed exponentiation, node $v$ may have to communicate with a large number of nodes in lower layers. To mitigate this issue, we intentionally perform only $\log \log n^\delta$ directed exponentiation steps, limiting the backwards communication to $n^\delta$ nodes and respecting the $O(n^\delta)$ total message bandwidth.

\subsection{MIS and Maximal Matching} \label{sec:boundmismm}

Maximal independent set and maximal matching algorithms follow from Theorem \ref{thm:bound3coloring}.

\begin{theorem} \label{thm:boundmis}
	There is a deterministic $O(\log \log n)$ time MIS algorithm for \emph{constant}-degree trees in the low-space \mpc model using $O(m)$ words of global memory.
\end{theorem}

\begin{proof}
	By Theorem \ref{thm:bound3coloring}, we can color the tree with 3 colors. For all colors $i$, perform the following. Nodes colored $i$ add themselves to the independent set, and all nodes adjacent to nodes colored $i$ remove themselves from the graph. Clearly this results in a maximal independent set in $O(1)$ time and the memory requirements are satisfied. 
\end{proof}

\begin{theorem} \label{thm:boundmm}
	There is a deterministic $O(\log \log n)$ time maximal matching algorithm for \emph{constant}-degree trees in the low-space \mpc model using $O(m)$ words of global memory.
\end{theorem}

\begin{proof}
	By Theorem \ref{thm:bound3coloring}, we can color the tree with 3 colors using a \rc decomposition. Recall that in the decomposition, each node $v \in V_i$ has at most two neighbors in subgraph $\bigcup V_j,~i\leq j \leq L$. Let us define the parent nodes of $v$. Node $u$ is a parent of $v$ if $u$ is a neighbor of $v$ and (i) $u$ belongs to a strictly higher layer than $v$ or (ii) $u$ belongs to the same layer and has a higher ID. Let us direct all edges from each node $v$ towards its parents. For all colors $i$, perform the following. Node $v$ colored $i$ proposes to its highest ID parent $u$, and $u$ accepts the proposal of the highest ID suitor; now either $v$ join the matching with $u$ or $u$ joins the matching with some other node. In the latter case, we repeat the same procedure with $v$'s other possible parent. Note that when a node joins the matching, it disables all other incident edges. As a result, all nodes colored $i$ have either joined the matching or they have no outward oriented edges. After iterating through all colors we are done, since all nodes have either joined the matching or they have no incident edges, implying that all their original neighbors belong to the matching. Clearly this results in a maximal matching in $O(1)$ time and the memory requirements are satisfied.
\end{proof}

\section{General Algorithm} \label{sec:genalgo}

In this section, we give an $O(\log^2 \log n)$ time deterministic algorithm for 3-coloring trees (of arbitrary degree) in the low-space \mpc model using $O(m)$ global memory. Similarly to the constant-degree algorithm, our general algorithm consists of a decomposition stage (Section \ref{sec:decompstage}) and a coloring stage (Section \ref{sec:coloringstage}). 

\begin{theorem} \label{thm:3coloring}
	There is a deterministic $O(\log^2 \log n)$ time 3-coloring algorithm for trees in the low-space \mpc model using $O(m)$ words of global memory.
\end{theorem}

\begin{proof}
Follows from Theorems \ref{thm:decompositioning} and \ref{thm:coloring}.
\end{proof}

\subsection{Decomposition Stage} \label{sec:decompstage}

We aim to compute a \rc decomposition of the input graph as defined in Definition \ref{def:rc} with $l=3$. Recall that a peeling step consist of one step of the \local algorithm of Section \ref{sec:bounddecompstage}, in which all degree-1 nodes and all nodes that belong to a path of length $\geq l$ place themselves into the current layer. When a node places itself into a layer, the node can be thought of as removed. When a node removes itself, it also disables all incident edges. When all nodes have removed themselves, the algorithm terminates. Consider the following \mpc algorithm.

\begin{enumerate}

\item Perform $\log \log n$ peeling steps.

\item For $i = 1,2,\dots,\log \log n + O(1)$ phases:

\begin{enumerate}
    \item Set $B_i \larr \min(n^\delta,\log^{2^{i-1}} n)$. Perform Steps (a)--(e) twice.
    \item All nodes partake in $r = \log \log n$ graph exponentiation steps, or as many as possible as long as their neighborhood is of size $\leq \sqrt{B_i}$.
    \item Now each node $v$ has performed some number of graph exponentiation steps, and knows its $2^k$-hop neighborhood ($k$ may be different for each node and $k \leq r$). Node $v$ simulates $2^k$ peeling steps $\log \log n + O(1)$ times. During each simulation: 
    \begin{itemize}
        \item If $v$ has successfully computed its own layer, $v$ first informs the nodes in its $2^k$-hop neighborhood of their layers and then removes itself from the graph. 
        \item If $v$ has failed to compute its own layer and some node(s) inform $v$ of its layer, $v$ accepts the assigned layer and removes itself from the graph. If no node has informed $v$ of its layer, $v$ does nothing.
    \end{itemize}
    \item Perform $O(1)$ peeling steps.
    \item All nodes wipe their memory.
\end{enumerate}
\end{enumerate}

Note that peeling in Step 1 is done in order to allocate enough global memory per node for the first phase. Peeling in Step 2d will play a crucial role in the analysis: it ensures that the root $v$ of subtree $T_i(v)$ removes itself in the proof of Lemma \ref{lemma:removetree}, and that enough nodes remove themselves in the proof of Lemma \ref{lemma:factordrop}.

\paragraph*{Simulating peeling steps.} Let us argue why we can simulate multiple peeling steps in one \mpc round in the general (arbitrary-degree) setting. Performing one peeling step in the \local model takes $l$ time. Hence, knowing your $x$-hop neighborhood in \mpc enables the simulation of $x/l$ peeling steps locally. Since $l$ is only a constant, and the term $1/l$ will dissipate in the big-$O$ notation of the runtime, we can pretend that knowing your $x$-hop neighborhood enables the simulation of exactly $x$ peeling steps without any effect on the algorithm. Note that as a result of simulating $x$ peeling steps, a node knows if it belongs to layer at most $x$. In addition to computing its own layer, $v$ can also compute if $u \in N^x(v)$ belongs to layer at most $x' \leq x$, if $N^{x'}(u) \subset N^x(v)$.

The previous simulation arguments hold both for nodes which did not stop partaking in exponentiation, as well as for nodes which did stop partaking in exponentiation, i.e., \textit{got stuck} in Step 2b. In Step 2c, all nodes perform $\log \log + O(1)$ simulations on their $2^k$-hop neighborhoods. If, as a result of a simulation, the layer of $v$ is

\begin{itemize}
    \item[] $\leq 2^k$: Node $v$ knows its exact layer and can remove itself. Before removing itself however, node $v$ informs nodes in $N^{2^k}(v)$ of their layers.
    \item[] $> 2^k$: Node $v$ does not know its layer. If some other node(s) $u$ informs $v$ of its layer, $v$ accepts the assigned layer and removes itself. Note that multiple nodes can inform $v$ of its layer. However, since all nodes simulate the same algorithm, they will inform the same layer.
\end{itemize}

\paragraph*{Correctness.} Informing neighborhood nodes of their layers is done in order to maintain the correctness of the decomposition. Consider node $v$ that has simulated $2^k$ peeling steps and knows its own layer. If all nodes in $N^{2^k}(v)$ place themselves in the layers simulated by $v$, the decomposition with regards to these nodes and $v$ will be correct.

Consider node $u$ such that it knows $N^x(u)$ and fails to compute its layer after a simulation step. Also, consider node $v$ that knows $N^y(v)$ and succeeds in computing its layer after a simulation step. If $v$ informs $u$ of its layer, it must be that $u \in N^y(v)$. Because $v$ knows the layer of $u$ and $u$ does not, it must be that $N^x(u) \subset N^y(v)$, since $v$ and $u$ both simulate the same algorithm. Therefore, if $u$ accepts the assigned layer from $v$, the decomposition with regards to $u$ and all other nodes in $N^x(u)$ will be correct.

\paragraph*{Technicality regarding layers.} In the algorithm description, we implicitly assume that all nodes are ``in sync'' regarding the layers that have been removed so far. This can be achieved in practice by assuming that in each simulation, we remove exactly $2^r=\log n$ layers, even if they are empty from the point of view of some nodes. Observe that each time a node is removed from the graph, the whole layer containing this node is removed. Also, each time a layer is removed, all lower layers are removed at the same time or earlier. Hence, by keeping track of variables $i$ and $r$, the aforementioned scheme is feasible. 

If no nodes get stuck in Step 2b, we will obtain a decomposition of size $O(\log n)$. If some nodes get stuck (because their $2^r$ neighborhoods contain $\omega(n^\delta)$ nodes), we will end up getting a slightly larger decomposition.

\paragraph*{Runtime analysis.} In the decomposition stage algorithm, parameter $B_i$ denotes the memory budget per node, i.e., an upper bound on the number of edges that fit into the memory of $v$ in phase $i$ such that we do not violate the local memory constraint nor linear global memory (this will become more apparent later). The idea behind a certain value $B_i$, is that we remove all at least $\sqrt{B_i}$ sized subtrees in each phase $i$. The reason we remove $\sqrt{B_i}$, and not $B_i$ sized subtrees is a memory management precaution related to graph exponentiation.

In the following analysis, let us imagine the graph as a rooted tree with all edges oriented towards some root. This ensures that each node has only one parent and the subtree $T_i(v)$ rooted at a node $v$ during phase $i$ is well defined. We emphasize that this orientation is used only for the analysis and we do not actually oriented any edges in the decomposition stage algorithm. During execution, the tree may shatter into a forest. In the following arguments we only address trees, but clearly everything also works in a forest. Let $T_i$ be a single connected component of what is left of the input graph $T$ in the beginning of phase $i$ and let $n_i=|T_i|$. We further present a number of lemmas in order to justify the runtime. 

\begin{figure}[H]
	\centering
	\includegraphics[width=0.79\textwidth]{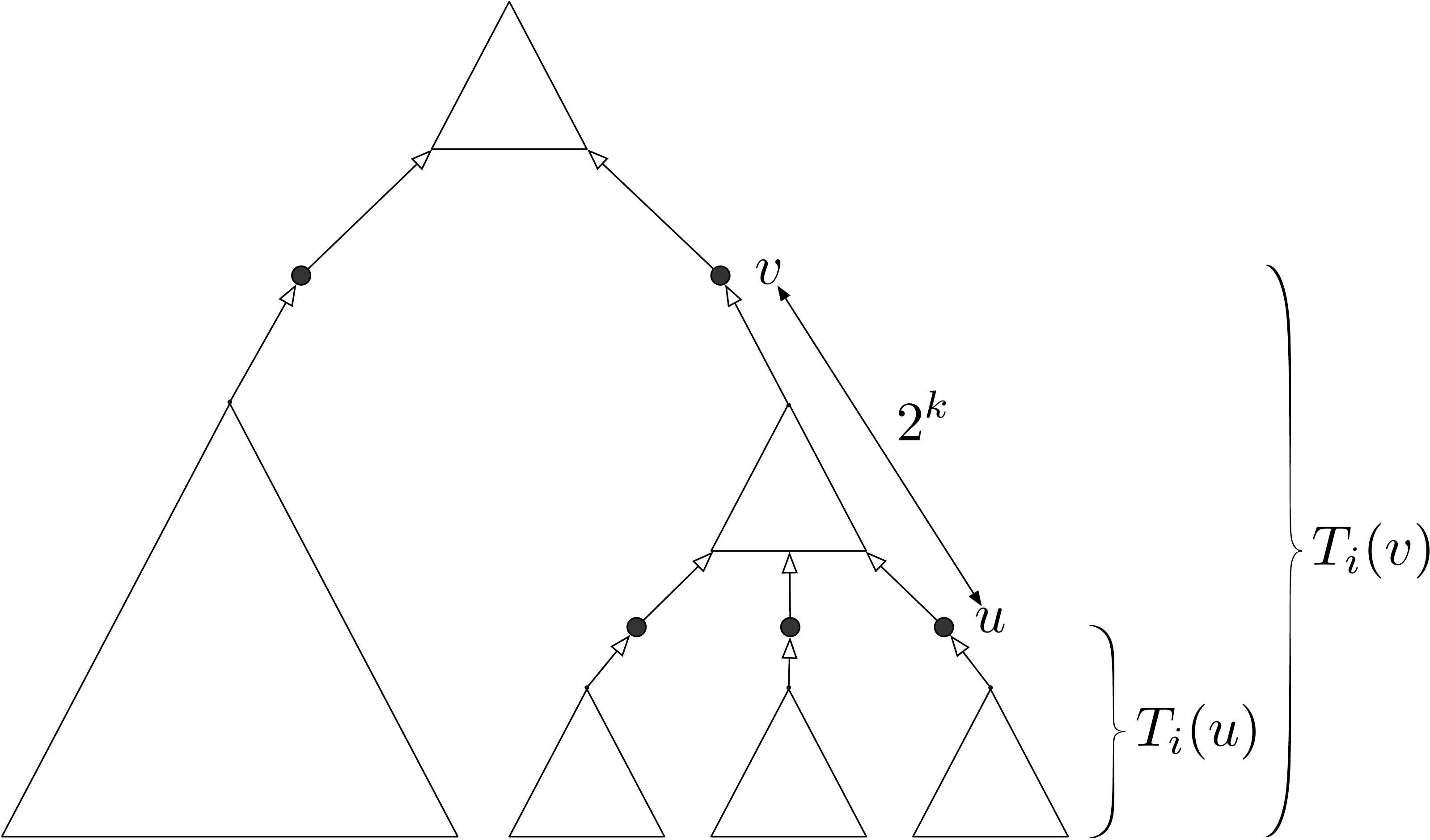}
	\caption{Subtree $T_i(v)$ rooted at node $v$ containing subtrees $T_i(u)$ rooted at nodes $u$.}
	\label{fig:ndeltaremoval}
\end{figure}

\begin{lemma} \label{lemma:enoughmemory}
	Consider $T_i(v)$ such that $|T_i(v)| \leq \sqrt{B_i}$. Nodes $u$ in $T_i(v)$ with distance at least $2^k$ from $v$ do not get stuck up to graph exponentiation step $k$.
\end{lemma}

\begin{proof}
Let $k$ be some arbitrary positive integer and consider any node $u$ in $T_i(v)$ with distance at least $2^k$ from $v$ as illustrated in Figure \ref{fig:ndeltaremoval}. The radii of neighborhoods of nodes in $T_i$ increase by a factor of at most 2 per step. Hence, after graph exponentiation step $k$, all nodes contained in $N^{2^k}(u)$ are actually also contained in $T_i(v)$. Since $|T_i(v)| \leq \sqrt{B_i}$, the claim follows.
\end{proof}

\begin{lemma} \label{lemma:removetree}
Subtree $T_i(v)$ gets removed during phase $i$ if $|T_i(v)| \leq \sqrt{B_i}$.
\end{lemma}

\begin{proof}
Let us first focus on subtrees $T_i(u)$ rooted at nodes $u$ with distance at least $2^{r}=\log n$ from $v$ in $T_i(v)$. If there are no such nodes, we can omit the current paragraph. Since $|T_i(v)| \leq \sqrt{B_i}$, it holds trivially that $|T_i(u)| \leq \sqrt{B_i}$. Also, according to Lemma \ref{lemma:enoughmemory}, nodes in $T_i(u)$ do not run out of memory during exponentiation and as a result, they will see their $2^r$-hop neighborhood. Since $r=\log \log n$, this results in $\log n$-hop neighborhood knowledge. As the size of the decomposition (if no nodes get stuck) is $O(\log n)$, we conclude that all $T_i(u)$ are removed in $O(1)$ simulations of Step 2c.

After a constant number of simulations, we are possibly left with some nodes in $T_i(v)$ that were not in any $T_i(u)$. Since we have no guarantee on the topology of $T_i$ outside $T_i(v)$, this may happen if, during exponentiation, these nodes get stuck. Since $r=\log \log n$, what is left of $T_i(v)$ has depth at most $\log n$. What is left of $T_i(v)$ is removed: after $\log \log n$ simulations of Step 2c, $T_i(v)$ should be empty, since by Claim \ref{lemma:killtrees} the depth of $T_i(v)$ decreases from $\log n$ by a factor of $4/3$ with each simulation. If the root $v$ of $T_i(v)$ is still present after the aforementioned local simulations, $v$ removes itself during the peeling steps of Step 2d at the latest, as $\deg(v)=1$.
\end{proof}

\begin{claim} \label{lemma:killtrees}
Let $d \leq \log n$ denote the depth of what is left of $T_i(v)$. Nodes in $T_i(v)$ with distance at least $3d/4$ from $v$ remove themselves during the next simulation. 
\end{claim}

\begin{proof}
Let $d' \larr \max \{2^k:2^k \leq d\}$. Since $d \leq \log n$ and $|T_i(v)| \leq \sqrt{B_i}$, nodes $w$ with distance at least $d-d'/2$ from $v$ see $T_i(w)$ and can determine the layer of each node in $T_i(w)$. Hence, nodes $w$ are successful in removing themselves in one simulation. Since $d' > d/2$, it holds that $d-d'/2 < 3d/4$ and the claim follows. 
\end{proof}

\begin{lemma} \label{lemma:factordrop}
In any arbitrary phase $i$, the number of nodes is reduced to at most $n_i/B_i$.
\end{lemma}

\begin{proof}
Let us identify all nodes $w$ that have a subtree $T_i(w)$ rooted at them such that $|T_i(w)| \leq B_i$, and all parent nodes $u$ of nodes $w$ that have a subtree $T_i(u)$ rooted at them such that $|T_i(u)| > B_i$. Let us color all nodes $w$ blue and all parent nodes $u$ black. Let us further identify nodes
\begin{align*}
	X = \{ v \in T_i : \text{$v$ is not blue and of degree $\geq 3$}\}
\end{align*}

and denote $x = |X|$. By Lemma \ref{lemma:removetree}, at least all subtrees of size $\leq \sqrt{B_i}$ are removed during phase $i$, i.e., all blue nodes are removed during phase $i$. It is straightforward to realize that after all blue nodes are removed, all leaf nodes are black. At the latest, during the peeling steps of Step 2d, all nodes except possibly nodes in $X$ and paths of length $<l$ connecting nodes in $X$ are left in $T_i$. Hence, we are left with at most $x +x(l-1)=xl$ nodes in $T_i$.

By Claim \ref{claim:corrleaf}, each node in $X$ corresponds to at least one leaf node, and hence to at least one black leaf node. Recall that all black nodes correspond to a subtree of size $>\sqrt{B_i}$ in $T_i$. Since all black leaf nodes are removed, at the latest, during the peeling steps of Step 2d, at least $x \sqrt{B_i}$ nodes are removed during phase $i$ and the following inequality must hold
\begin{align*}
    x \sqrt{B_i} \leq n_i ~~\Leftrightarrow~~  x \leq \frac{n_i}{\sqrt{B_i}} ~~\Leftrightarrow~~ x l \leq \frac{ln_i}{\sqrt{B_i}} ~~\Leftrightarrow~~ n_{i+1} \leq \frac{ln_i}{\sqrt{B_i}}.
\end{align*}

After some additional peeling steps of Step 2d, and the fact that we perform Steps 2(b)--(e) twice, we can conclude that $n_{i+1} \leq n_i / B_i$.
\end{proof}

\begin{claim} \label{claim:corrleaf}
	In trees, each node of degree at least 3 corresponds to at least one degree-1 node. 
\end{claim}

\begin{proof}
	Consider rooting the tree such that it has depth $d$. Let us iterate though all depths and match all nodes of degree $\geq 3$ with one leaf node. We omit depth $d$, since there are no nodes of degree at least 3 at depth $d$. Consider any node $v$ at depth $d-1$ of degree $\geq 3$. Since it has at least two children and hence at least two descendant leaf nodes, it can match with one of them, and leave the other one be. This argument can be applied for each depth until all nodes of degree $\geq 3$ are matched with one of their descendant leaf nodes.
\end{proof}

\begin{lemma} \label{lemma:decompingruntime}
The decomposition stage algorithm terminates in at most $\log \log n + O(1)$ phases.
\end{lemma}

\begin{proof}
Recall that in the beginning of each phase $i$, we set $B_i \larr \min(n^\delta,\log^{2^{i-1}} n)$. Note that we can assume that $n^\delta > \log n$. In some phase $k< \log \log n$ it holds that $\log^{2^{k-1}}n > n^\delta$, since $\log^{\log n}n = n^{\omega(1)}$. Hence, it holds that $B_i=n^\delta$ for all following phases $i>k$.

The are at most $n$ nodes left after $k$ phases, and by Lemma \ref{lemma:factordrop}, there are at most $n/n^{j\delta}$ nodes left after $j$ additional phases. Since $\delta$ is a constant, so is $j$ and the claim follows.
\end{proof}

\begin{theorem} \label{thm:decompositioning}
The decomposition stage algorithm gives a \rc decomposition of size $O(\log^2 \log n \cdot \log n)$ in $O(\log^2 \log n)$ time in the low-space \mpc model using $O(m)$ words of global memory.
\end{theorem}

\begin{proof}
The peeling steps of Step 1 remove $\log \log n$ layers. In each following simulation, $2^r=\log n$ layers are removed. Since $\log \log n + O(1)$ simulations are performed during each phase and, by Lemma \ref{lemma:decompingruntime}, there are at most $\log \log n + O(1)$ phases, it follows that the decomposition is of size $O(\log^2 \log n \cdot \log n)$.

The peeling steps of Step 1 take $\log \log n$ time.  Also, observe that each phase of Step 2 consist of $\log \log n$ exponentiation steps, $\log \log n + O(1)$ simulations, a constant number of peeling steps and memory cleanup. By Lemma \ref{lemma:decompingruntime} there are at most $\log \log n + O(1)$ phases, and hence the total runtime is $O(\log^2 \log n)$.

Now for the memory analysis. The decomposition stage algorithm never violates $O(n^\delta)$ local memory, since Step 1 does not require memory, and $B_i \leq n^\delta$ holds for all phases $i$ in Step 2. It is left to argue about linear global memory. Again, observe that Step 1 do not require memory. Since each peeling step removes a constant fraction of nodes, at most $n/\log n$ nodes are left after Step 1. In other words, we have allocated $\log n$ global memory for each node for the first phase of Step 2. Since all nodes wipe their memory at the end of each phase, memory usage can be analyzed for each phase separately. Recall that by Lemma \ref{lemma:factordrop}, at the end of each phase, it holds that $n_{i+1} \leq n_{i}/B_{i}$.

In the start of the first phase, there are at most $n_1 = n/\log n = n/B_1$ nodes in the graph. During the first phase, nodes use at most $B_1$ local memory, and after the first phase, there are at most
\begin{align*}
	n_2 \leq n_1 / B_1 = n / B_1^2 = n / B_2
\end{align*}

nodes left. Similarly, after the second phase, there are at most 
\begin{align*}
	n_3 \leq n_2 / B_2 = n / B_2^2 = n / B_3
\end{align*}

nodes left. Hence, in the beginning of each phase $i$ it holds that $n_i \leq n / B_i$. Since each node during phase $i$ uses at most $B_i$ local memory, this results in $n_i B_i \leq n$ global memory use.

During the algorithm, it holds that $n_{i} \leq n/B_i =  n/ \log^{2^{i-1}} n$ until the first phase $k< \log \log n$ such that $\log^{2^{k-1}}n > n^\delta$. Such $k$ always exists, since $\log^{\log n}n = n^{\omega(1)}$. For all following phases $i>k$ it holds that $n_i \leq n/B_i$ and $B_i=n^\delta$. We conclude that global memory use is $n_i B_i \leq n$ for all phases in the algorithm.
\end{proof}

\subsection{Coloring Stage} \label{sec:coloringstage}

Consider having the \rc decomposition described in Definition \ref{def:rc} with $l=3$, which we get from the previous section. Consider performing the same pre-processing steps and the same algorithm as in the constant-degree case of Section \ref{sec:boundcoloringstage}. Our setting is identical apart from two key aspects: (i) the decomposition is now of size $O(\log^2 \log n \cdot \log n)$ and (ii) the degree of each node is unbounded. We address the issues that arise from (i) and (ii) with the following modifications to the algorithm.

\begin{itemize}
	\item[(i)] The larger size of our decomposition only plays a role in Step 3, where we now have to perform $O(\log^2 \log)$ simulation instead of $O(1/\delta)$.
	\item[(ii)] Recall that in the directed exponentiation of the algorithm, even though each node $v$ depends on at most one node, multiple nodes may depend on $v$. After $\log \log n^\delta$ directed exponentiation steps, node $v$ may have to communicate with $\Delta^{ \delta \log n}$ nodes in lower layers. Since $\Delta$ is now unbounded, this communication may not be feasible in the low-space \mpc model. We fix this issue with an additional load balancing pre-processing step. Each machine hosting node $v$ that participates in directed exponentiation and has $\omega(1)$ ingoing edges simulates several virtual machines that are responsible for $O(1)$ incoming edges of $v$. Note that these virtual machines are responsible for different edges of $v$, and all need to keep a copy of the single outgoing edge of $v$. 	
\end{itemize}

\begin{theorem} \label{thm:coloring}
	The coloring stage algorithm terminates in $O(\log^2 \log n)$ time in the low-space \mpc model using $O(m)$ global memory.
\end{theorem}

\begin{proof}
	Correctness follows directly from Section \ref{sec:boundcoloringstage}. The runtime is otherwise identical, except for Step 3, where we perform $O(\log^2 \log n)$ simulations instead of $O(1/\delta)$ as discussed in (i). Local memory restrictions during directed exponentiation are respected, since similarly to the constant-degree setting, each node depends on at most one node, forming dependency paths. In order to not exceed total message bandwidth, we perform load balancing as described in (ii) and each node (or more precisely each virtual machine hosting a constant number of ingoing edges of some node) has constant indegree. Note that when performing load balancing, we store multiple copies of the outgoing edge per node. However, it is simple to realize that we always store at most one copy of an outgoing edge per ingoing edge. Hence, linear global memory is also respected.
\end{proof}

\subsection{MIS and Maximal Matching} \label{sec:mismm}

Similarly to Section \ref{sec:boundmismm}, we get the following $O(\log^2 \log n)$ time algorithms.

\begin{theorem} \label{thm:mis}
There is a deterministic $O(\log^2 \log n)$ time MIS algorithm for trees in the low-space \mpc model using $O(m)$ words of global memory.
\end{theorem}

\begin{proof}
Follows from Theorem \ref{thm:3coloring} and the proof of Theorem \ref{thm:boundmis}.
\end{proof}

\begin{theorem} \label{thm:mm}
There is a deterministic $O(\log^2 \log n)$ time maximal matching algorithm for trees in the low-space \mpc model using $O(m)$ words of global memory.
\end{theorem}

\begin{proof}
Follows from Theorems \ref{thm:3coloring} and \ref{thm:decompositioning}, and the proof of Theorem \ref{thm:boundmm}.
\end{proof}

\clearpage
\phantomsection
\bibliographystyle{alphaurl}
\bibliography{coloring-trees}

\end{document}